\documentclass[letterpaper, 10 pt, conference]{ieeeconf}  

\IEEEoverridecommandlockouts                              
\usepackage{amsmath,amssymb,amsthm,bm}
\usepackage{graphicx}
\usepackage{booktabs}
\usepackage{array}
\usepackage{pifont}

\usepackage{algorithm}
\usepackage{algpseudocode}
\usepackage{xcolor}
\makeatletter
\def\thm@space@setup{\thm@preskip=1pt \thm@postskip=1pt}
\makeatother

\makeatletter
\long\def\@makecaption#1#2{%
  \ifx\@captype\@IEEEtablestring
    \vskip 5pt plus 1pt minus 1pt
    \setbox\@tempboxa\hbox{\footnotesize #1.~~ #2}%
    \ifdim \wd\@tempboxa >\hsize
      \setbox\@tempboxa\hbox{\footnotesize #1.~~ }%
      \parbox[t]{\hsize}{\footnotesize\noindent\unhbox\@tempboxa#2}%
    \else
      \hbox to\hsize{\footnotesize\box\@tempboxa\hfil}%
    \fi
    \vskip 1pt
  \else
    \@IEEEfigurecaptionsepspace
    \setbox\@tempboxa\hbox{\footnotesize #1.~~ #2}%
    \ifdim \wd\@tempboxa >\hsize
      \setbox\@tempboxa\hbox{\footnotesize #1.~~ }%
      \parbox[t]{\hsize}{\footnotesize\noindent\unhbox\@tempboxa#2}%
    \else
      \hbox to\hsize{\footnotesize\hfil\box\@tempboxa\hfil}%
    \fi
  \fi}
\makeatother

\newcommand{\Qbar}{\bar{Q}}
\newcommand{\codeurl}{https://github.com/SAILRIT/Agent-shield-certificates}
\usepackage[noadjust]{cite}
\usepackage{microtype}
\PassOptionsToPackage{hyphens}{url}
\usepackage[colorlinks=true,linkcolor=blue,citecolor=blue,urlcolor=blue]{hyperref}
\hypersetup{
  pdftitle={Controller-Agnostic Runtime Safety Certificates for Tool-Using AI
            Agents via Data-Driven Robust Reachability},
  pdfauthor={Ali Baheri; Peng Wei},
  pdfsubject={Data-driven safety certification for tool-using AI agents},
  pdfkeywords={data-driven control; learning for control; Markov processes; reachability analysis; safety filters},
  pdflang={en},
  pdfdisplaydoctitle=true}

\theoremstyle{plain}
\newtheorem{theorem}{Theorem}
\newtheorem{proposition}{Proposition}
\newtheorem{lemma}{Lemma}
\newtheorem{corollary}{Corollary}
\theoremstyle{definition}
\newtheorem{assumption}{Assumption}

\newcommand{\Sset}{\mathcal{S}}
\newcommand{\Aset}{\mathcal{A}}
\newcommand{\Su}{\mathcal{S}_{\mathrm{u}}}
\newcommand{\Dshield}{\mathcal{D}}
\newcommand{\Cset}{\mathcal{C}}
\newcommand{\Phat}{\widehat{P}}
\newcommand{\Vbar}{\bar{V}}
\newcommand{\Vstar}{V^{\star}}
\newcommand{\simplex}{\Delta(\Sset)}
\newcommand{\halt}{\textsc{halt}}
\newcommand{\inner}[2]{\langle #1,\, #2\rangle}
\newcommand{\one}{\mathbf{1}}
\newcommand{\Prob}{\mathbb{P}}
\newcommand{\Ind}[1]{\mathbf{1}\!\left[#1\right]}

\makeatletter
\let\oldthebibliography\thebibliography
\renewcommand{\thebibliography}[1]{\oldthebibliography{#1}%
  \setlength{\itemsep}{0pt}\setlength{\parsep}{0pt}\setlength{\topsep}{0pt}\setlength{\partopsep}{0pt}}
\makeatother

\makeatletter
\def\ps@preprint{\def\@oddhead{}\def\@evenhead{}%
  \def\@oddfoot{\hfil\footnotesize\thepage\hfil}\def\@evenfoot{\hfil\footnotesize\thepage\hfil}}
\makeatother
\begin{document}

\title{\LARGE \bf Controller-Agnostic Runtime Safety Certificates for Tool-Using AI Agents \\ via Data-Driven Robust Reachability}

\author{Ali Baheri$^{1}$ and Peng Wei$^{2}$%
\thanks{$^{1}$Ali Baheri is with the Department of Mechanical Engineering,
Kate Gleason College of Engineering, Rochester Institute of Technology,
Rochester, NY, USA.
{\tt\small akbeme@rit.edu}}%
\thanks{$^{2}$Peng Wei is with the Department of Mechanical and Aerospace Engineering,
School of Engineering and Applied Science, George Washington University,
Washington, DC, USA.
{\tt\small pwei@gwu.edu}}%
}

\maketitle
\thispagestyle{empty}

\begin{abstract}
Software agents built on large language models act by issuing tool calls, some of which cannot be undone. To control theory the model is an opaque feedback controller, replaced every few months, whose closed loop cannot be analyzed directly. We propose a runtime filter, or shield, between any such controller and its environment that certifies a bound on the probability of reaching an unsafe condition during a bounded task. Given a predicate abstraction and candidate successor sets, everything else comes from logs. The loop becomes a finite controlled Markov chain, the estimated kernel is surrounded by a distribution-free confidence set, and robust reach--avoid value iteration runs over it. The certificate is probably approximately correct and uniform over controllers respecting the filter, so one model may replace another without recomputation. In all tested datasets the true risk stays within certificate and budget; in a file-system sandbox the filter stays sound against a hazard absent from the logs, where a plug-in baseline certifies safety and fails; and it holds seven controllers, six built on open-weights models, to zero unsafe outcomes in 1,050 episodes.
\end{abstract}


\section{Introduction}

A new kind of autonomous system is asking for the control community's attention. Agents built on large language models~\cite{yao2023react,schick2023toolformer} take a goal in plain language and pursue it by picking a tool call, seeing what comes back and picking again until done. To a control engineer, the model is a feedback controller and the surrounding software the plant.

Two things make the loop awkward for the classical toolbox. First, some actions are irreversible. A deleted database or a wired payment cannot be taken back, so what must be bounded is the probability one is ever triggered before the task ends, a finite-horizon reach--avoid requirement. Second, the controller is opaque and replaced within months, so any argument tied to its identity is obsolete almost as soon as written.

Together they rule out most existing machinery. Barrier-function and predictive safety filters~\cite{ames2019cbf,wabersich2021safetyfilter,baheri2026ensemble} and Hamilton--Jacobi reachability~\cite{bansal2017hj} start from a dynamical model, supervisory control~\cite{ramadge1987supervisory} gives only logical guarantees, and policy-specific verification must be redone when the controller changes~\cite{alshiekh2018shielding}. What agents do offer is data. They all log state, action and outcome, inviting a data-driven treatment as in~\cite{depersis2020data}.

So we ask whether any LLM agent can be wrapped in a thin filter, built from logs with no known transition probabilities or training, guaranteeing a chosen ceiling on the probability of an unsafe outcome within a bounded task, certified once for every controller respecting it. Given predicates making the loop Markov and each call's possible outcomes, it can.

\noindent\textbf{Contributions.} This letter makes the following contributions.
\begin{enumerate}
\item \emph{A plant model for an agent.} A few engineer-chosen Boolean facts turn the agent loop into a finite controlled Markov chain with a reach--avoid specification (Sec.~\ref{sec:setup}).
\item \emph{A certificate from ordinary logs.} A robust reach--avoid bound on the risk of irreversible harm, computed over a distribution-free confidence set, is valid with chosen confidence for \emph{every} controller the filter admits (Sec.~\ref{sec:cert}).
\item \emph{A synthesis procedure.} It seeks the least restrictive filter meeting the risk budget (Sec.~\ref{sec:synth}).
\item \emph{Validation.} We check the certificate against ground truth, compare our filter with four alternatives, and test one filter, built from a scripted agent's logs, with that agent and three open-weights language models~\cite{qwen2024technical} (Sec.~\ref{sec:exp}).
\end{enumerate}
\noindent Proofs of all results are given in the Appendix.

\noindent\textbf{Relation to data-driven robust MDPs and shielding.}
Suilen \emph{et al.}~\cite{suilen2022robust} maintain Bayesian intervals on an uncertain MDP; Badings \emph{et al.}~\cite{badings2023robust} lift PAC intervals to an interval MDP; Reed and Lahijanian~\cite{reed2025shielding} build one by deep-kernel regression, computing a maximally permissive safe-policy set. Closest, Hamel-De le Court \emph{et al.}~\cite{hameldelecourt2026robust} shield robust MDPs against a safety-LTL threshold; on the RMDP their shield is sound and optimal, admitting exactly the safe policies. We claim no novelty in uniformity over controllers, which shielding supplies by construction. What differs is the data, as most acquire samples through a scheme they control, and where a passive stream is admitted~\cite{suilen2022robust} the uncertainty is Bayesian, not distribution-free. Ours is a narrower, distribution-free construction for \emph{ordered, passively collected} traces with adaptive, uncontrolled visitation and explicit vacuous treatment of under-covered rows, feeding a robust finite-horizon shield. Ray~\cite{ray2026enforced} concurrently asks what a gate can enforce at all,
a question of enforcement rather than certification.
Table~\ref{tab:related} sets these beside this letter.

\begin{table*}[t]
\centering\footnotesize
\setlength{\tabcolsep}{4pt}
\scriptsize\setlength{\tabcolsep}{1.7pt}
\begin{tabular}{@{}llllll@{}}
\toprule
Work & Data & Uncertainty set & Guarantee ranges over & Spec. & Guarantee \\
\midrule
Suilen \emph{et al.}~\cite{suilen2022robust}
 & interactive or streamed & Bayesian intervals & the synthesized policy & temporal logic & credible \\
Badings \emph{et al.}~\cite{badings2023robust}
 & prescribed samples & PAC intervals (iMDP) & the synthesized policy & reach--avoid & PAC \\
Reed and Lahijanian~\cite{reed2025shielding}
 & queried samples & kernel intervals (iMDP) & a maximal safe-policy set & safety & probabilistic \\
Hamel-De le Court \emph{et al.}~\cite{hameldelecourt2026robust}
 & simulator, per-pair budget & Clopper--Pearson & every compliant policy; all safe admitted & safety LTL & PAC, sound + optimal \\
\midrule
This letter & \textbf{passive, adaptively visited logs} & distribution-free $L_1$ balls
 & every admissible controller & reach--avoid, \halt & PAC, distribution-free \\
\bottomrule
\end{tabular}
\caption{\textbf{This letter and the closest prior work.} Ours differs mainly in its passive, adaptively visited logs and distribution-free uncertainty set.}
\label{tab:related}
\end{table*}

\section{The Agent--Environment Loop as a Controlled Markov Chain}
\label{sec:setup}

This section models the agent loop as a finite controlled Markov chain with
a reach--avoid specification.

\subsection{From transcripts to a finite state machine}
The controller is a language model whose natural state is the whole transcript $h_t\in\mathcal{H}$, every prompt, call and reply so far, an unbounded object. We replace it with a short
list of Boolean \emph{predicates} chosen for the application, such as ``holds
write capability,'' ``inside a protected directory,'' ``an irreversible call has
executed.'' The map $\varphi:\mathcal{H}\to\Sset$ sends a transcript to the
truth values of these predicates, so the abstract state $s_t=\varphi(h_t)$
lives in a finite set with $|\Sset|=n$. Actions $a\in\Aset$ are a finite
catalog of tool-call \emph{types}, $|\Aset|=m$, one of which, $\halt$, hands
control to a human and ends the episode safely. An action may refine a tool
name by safety-relevant argument classes; any dependence left out of $a$ must
appear in the abstract state instead.

\begin{assumption}[Abstract Markovness]\label{ass:markov}
There is a kernel $P$ with $\Prob[s_{t+1}=s'\mid h_t,a_t]=P(s'\mid s_t,a_t)$, $s_t=\varphi(h_t)$, at every $t$ and for \emph{any} history-dependent $a_t$. Hence, given the current abstract state and action, the next is independent of the rest of $h_t$, with a kernel independent of $t$.
\end{assumption}

That is, once the predicates and the action are known, the rest of the transcript says nothing more about where the system goes next, even though the controller may have used it to choose the action. The analysis rests on this assumption, so the predicates must be rich enough to make it hold; Experiment~4 tests it on logs, and Prop.~\ref{prop:eta} bounds the cost of a violation.

\begin{assumption}[Irreversibility]\label{ass:absorb}
There is an unsafe set $\Su\subset\Sset$ (a harmful/irreversible call has occurred); each $s\in\Su$ is absorbing, $P(s\mid s,a)=1$ for all $a$.
\end{assumption}

\begin{assumption}[Safe stop]\label{ass:halt}
A distinguished sink $\bot\in\Sset$ is absorbing and safe. From every non-unsafe, non-sink state, $\halt$ leads to it,
\[
P(\bot\mid s,\halt)=1\qquad\text{for } s\notin\Su\cup\{\bot\}.
\]
(Unsafe states remain absorbing under \emph{every} action, including $\halt$, by A\ref{ass:absorb}.)
\end{assumption}

\begin{assumption}[Controller-independence]\label{ass:indep}
$P$ depends only on $(s,a)$, not on which controller selected $a$.
\end{assumption}

A\ref{ass:absorb} and A\ref{ass:halt} say only that harm cannot be undone
and stopping is always possible. A\ref{ass:indep} lets one certificate serve every controller, since a file system answers ``delete'' the same way whoever issued it, provided the action or the predicates record every argument that matters, such as the path. Environments that adapt to the agent's identity are excluded.

\noindent\textbf{Controller, shield, admissibility.} A \emph{controller}
is any nonanticipative rule for the next action, a family of kernels
$\pi_t(\cdot\mid h_t)\in\Delta(\Aset)$ that may depend on the entire transcript,
vary with time and randomize. A language model is one. A \emph{shield}
is a map $\Dshield:\Sset\to2^{\Aset}$ listing the actions permitted in each
abstract state, always including $\halt$. A controller is \emph{admissible},
written ${\pi\rhd\Dshield}$, if $\operatorname{supp}{\pi_t(\cdot\mid h_t)}\subseteq\Dshield(\varphi(h_t))$ for all $t$, that is, if every action it \emph{executes} is permitted. The shield constrains only what is executed; how a forbidden proposal is handled is a separate choice, and Sec.~\ref{sec:exp} uses three, namely restricting the agent to permitted actions (the random and adversarial controllers of Experiment~1), replacing the proposal by $\halt$, which ends the episode (its competent controller), and blocking it and asking again, with $\halt$ on the fourth consecutive refusal (the language models of Experiment~3). Each makes the agent plus its rule a nonanticipative controller
whose executed action lies in $\Dshield(\varphi(h_t))$. A blocked proposal
causes no transition, does not advance the step count, which counts executed
calls, and leaves $\varphi(h_t)$ unchanged, as the predicates record facts about
the environment. That composite is admissible and is the $\pi$ every result
below concerns, immediate $\halt$ substitution being the case of no second
attempt. The shield is thus a supervisory
controller~\cite{ramadge1987supervisory,alshiekh2018shielding} between a
stochastic plant and a black-box decision maker.

\subsection{The finite-horizon reach--avoid specification}
The quantity we certify is the probability that the loop ever reaches an
unsafe state within the task horizon. Fix a horizon $T$ and a start $h_0$, $s_0=\varphi(h_0)$. For admissible $\pi$ and kernel $P$ the \emph{unsafe-reach probability} is
\begin{equation}\label{eq:reachprob}
V^{\pi}_T(s_0;P) \;:=\; \Prob\!\left[\,\exists\, t\le T:\; s_t\in\Su \,\middle|\, h_0,\pi,P\right].
\end{equation}
Since unsafe states are absorbing this is the probability of \emph{ever}
entering $\Su$ within $T$ steps. Not knowing which controller will run, we
control the worst case over all of them,
\begin{equation}\label{eq:worstcase}
\Vstar_T(s;P,\Dshield) \;:=\; \sup_{\pi\,\rhd\,\Dshield} V^{\pi}_T(s;P),
\end{equation}
a supremum over every nonanticipative admissible controller, which
collapses to a state-feedback recursion.

\begin{lemma}[Worst-case reach--avoid value]\label{lem:worstcase}
Under A\ref{ass:markov}--A\ref{ass:halt}, for every $s$ and $k$ the worst-case value $\Vstar_k(s;P,\Dshield)$ is produced by $\Vstar_0=\one_{\Su}$ and, for $s\notin\Su\cup\{\bot\}$,
\begin{equation}\label{eq:truevi}
\Vstar_k(s)=\max_{a\in\Dshield(s)}\ \sum_{s'}P(s'\mid s,a)\,\Vstar_{k-1}(s'),
\end{equation}
with $\Vstar_k\!\equiv\!1$ on $\Su$, $\Vstar_k(\bot)\!=\!0$. The supremum in~\eqref{eq:worstcase} is attained by a \emph{deterministic Markov} policy $a^\star_t(s)\in\arg\max_{a\in\Dshield(s)}\inner{P_a(s)}{\Vstar_{T-t-1}}$, $P_a(s):=P(\cdot\mid s,a)$.
\end{lemma}

Equation~\eqref{eq:truevi} is ordinary reach--avoid value
iteration~\cite{summers2010verification}. At each state the worst controller
picks the permitted action with the highest chance of eventual harm, and no
history-dependent or randomized controller can do worse. With $P$ known this
would finish the problem. But $P$ is unknown, and the rest of this letter certifies an upper bound on~\eqref{eq:worstcase} from logs.

\section{Data-Driven Robust Reach--Avoid Certificate}
\label{sec:cert}

We replace the unknown kernel by a confidence set built from the logs, the analogue of an uncertainty set in robust control, and run value iteration against its worst kernel; the result bounds~\eqref{eq:worstcase} with high probability.

\subsection{Estimating the kernel from adaptively collected logs}
The data are $M$ logged episodes of abstract state, action and next state. Each pair $(s,a)$, a \emph{row} of the transition matrix $P_a$, is visited $N(s,a)$ times; $X^{s,a}_j$ is the successor recorded at the $j$-th visit, and the natural estimate of $P(\cdot\mid s,a)$ is their empirical distribution. The difficulty is that the logs are collected in closed loop. As in closed-loop identification, which rows are visited, and how often, depends on earlier outcomes. Hence $N(s,a)$ is random and correlated with the very successors being averaged, and a confidence bound for a sample of fixed size cannot simply be evaluated at $N(s,a)$.

We avoid this by construction. Before any data are seen, the protocol commits to a per-row budget $n_0(s,a)$, and the estimator uses \emph{exactly} the first $n_0(s,a)$ successors at each row, discarding the surplus. Thus $\Phat(s'\mid s,a)=N_0(s,a,s')/n_0(s,a)$ with truncated counts $N_0(s,a,s')=\sum_{j\le n_0}\one[X^{s,a}_j=s']$, and $\Phat_a(s):=\Phat(\cdot\mid s,a)$. The count $N(s,a)$ only shows whether the budget was met; a row that falls short takes the vacuous set below, which contains every kernel, so a missed budget loosens the certificate but never invalidates it. The budget must be fixed from design constants before any log is examined, since choosing it after seeing the counts would reintroduce that dependence; for a retrospective log a time-uniform confidence sequence or a holdout is safer.

We also assume a mild structural prior. For each row a \emph{candidate
successor set} $K_{s,a}\supseteq\operatorname{supp}P(\cdot\mid s,a)$ of size
$k_{s,a}$ is known, meaning what a tool call can return, though not with what
probability. Enlarging $K_{s,a}$ is always safe, since the confidence set only
grows and still contains the truth, whereas omitting a possible successor breaks
the guarantee. The size of $K_{s,a}$ affects only how tight the certificate is. The radius grows like $\sqrt{k_{s,a}/n_0}$ with the number of candidate successors,
since $\ln(2^{k_{s,a}}-2)\approx k_{s,a}\ln2$, and $K_{s,a}=\Sset$ roughly doubles it on our grid.

For $N$ i.i.d.\ samples on $k$ outcomes, Weissman \emph{et al.}~\cite{weissman2003inequalities} give an $L_1$ analogue of Hoeffding's inequality,
\begin{equation}\label{eq:weissman}
\Prob\!\left[\,\|\Phat_a(s)-P_a(s)\|_1\ge\eta\,\right]\le (2^{k}-2)\,e^{-N\eta^2/2}.
\end{equation}
Setting the right-hand side to $\delta/R$ (a union bound over the $R$ rows) and solving for $\eta$ gives the per-row radius
\begin{equation}\label{eq:beta}
\beta_{s,a}=\sqrt{\tfrac{2}{n_0(s,a)}\,\ln\!\Big(\tfrac{(2^{k_{s,a}}-2)\,R}{\delta}\Big)},
\end{equation}
and, for a stochastic row (one with $k_{s,a}\ge2$ candidate successors), the confidence set (an $L_1$ ball intersected with the simplex)
\begin{equation}\label{eq:cset}
\begin{aligned}
\Cset_{s,a}=\big\{q\in\simplex:\ &\|q-\Phat_a(s)\|_1\le\beta_{s,a},\\[-1pt]
&\operatorname{supp}q\subseteq K_{s,a}\big\}.
\end{aligned}
\end{equation}
A \emph{deterministic} row, $k_{s,a}=1$, has nothing to learn. The factor
$2^{k_{s,a}}-2$ vanishes, $\beta_{s,a}=0$ and $\Cset_{s,a}=\{\delta_{s'}\}$, as
for $\halt$. A row with $N(s,a)<n_0(s,a)$ takes the vacuous set $\Delta(K_{s,a})$.

Inequality~\eqref{eq:weissman} assumes a fixed sample size. The next lemma, the core of the construction, shows that the committed budget keeps~\eqref{eq:beta} valid for closed-loop logs. With probability at least $1-\delta$, every confidence set contains its true row (event $E$). The lemma also covers an approximately Markov abstraction.

\begin{lemma}[Uniform validity]\label{lem:valid}
Let $Q_j$, supported in $K_{s,a}$, be the conditional law of the successor at the $j$-th visit to
$(s,a)$ given the past and $\Qbar_a(s):=n_0^{-1}\sum_{j\le n_0}Q_j$. With
probability at least $1-\delta$ over the draw of the logs,
$\Qbar_a(s)\in\Cset_{s,a}$ for \emph{all} $(s,a)$ at once. Under
A\ref{ass:markov} every $Q_j$ equals $P_a(s)$, so this reads
$P_a(s)\in\Cset_{s,a}$.
\end{lemma}

\noindent\emph{Proof sketch} (see the Appendix). Whether a visit is among the first $n_0$ is fixed before its successor is drawn, so for each $A\subseteq K_{s,a}$ the terms $\one[X^{s,a}_j\in A]-Q_j(A)$ are martingale differences of range one; Azuma--Hoeffding and a union bound over the nontrivial sets $A$ and the rows give~\eqref{eq:beta} without independence.
\subsection{Robust value iteration}
Lemma~\ref{lem:worstcase} took the worst case over controllers for a known
kernel. We now also take the worst case over the kernels the data cannot rule
out. Define $\Vbar_0=\one_{\Su}$ and, for $s\notin\Su\cup\{\bot\}$,
\begin{equation}\label{eq:robustvi}
\Vbar_k(s)=\max_{a\in\Dshield(s)}\ \sigma\!\big(\Cset_{s,a};\,\Vbar_{k-1}\big),
\end{equation}
with $\Vbar_k\!\equiv\!1$ on $\Su$, $\Vbar_k(\bot)=0$, and the support function
\begin{equation}\label{eq:support}
\sigma(\Cset_{s,a};V):=\max_{q\in\Cset_{s,a}}\ \inner{q}{V},
\end{equation}
the largest one-step value any kernel in the confidence set could
produce. Equation~\eqref{eq:robustvi} is robust dynamic programming~\cite{iyengar2005robust} for reach--avoid; its only nontrivial step,~\eqref{eq:support}, is a linear program over an $L_1$ ball $\cap$ simplex, solved in closed form by shifting mass toward the riskiest successors.

\begin{proposition}[Inner oracle]\label{prop:oracle}
Let $\hat p=\Phat_a(s)$ and order the support so $V$ is nonincreasing. Then~\eqref{eq:support} equals $\inner{\hat p}{V}$ plus the gain from moving at most $\min(\beta_{s,a}/2,1)$ of mass from the smallest-$V$ entries of $\hat p$ to the largest, capped by donor and receiver capacity and stopped when no strictly value-increasing transfer remains; the moved mass may be zero, as at $\hat p=(1,0)$, $V=(1,0)$. It costs $\mathcal{O}(k_{s,a}\log k_{s,a})$.
\end{proposition}

\subsection{The certificate}
\begin{theorem}[PAC controller-agnostic safety]\label{thm:main}
Under A\ref{ass:markov}--A\ref{ass:indep}, the support condition $\operatorname{supp}P(\cdot\mid s,a)\subseteq K_{s,a}$ and the sets~\eqref{eq:cset}, with probability at least $1-\delta$ over the draw of the logs, for every $s$, every $k\le T$ and \emph{every} admissible controller $\pi\rhd\Dshield$,
\begin{equation}\label{eq:sandwich}
V^{\pi}_k(s;P)\ \le\ \Vstar_k(s;P,\Dshield)\ \le\ \Vbar_k(s).
\end{equation}
In particular, if the shield satisfies $\Vbar_T(s_0)\le\varepsilon$, then
\begin{equation}\label{eq:guarantee}
\Prob_{\mathrm{data}}\!\Big[\ \sup_{\pi\,\rhd\,\Dshield} V^{\pi}_T(s_0;P)\ \le\ \varepsilon\ \Big]\ \ge\ 1-\delta .
\end{equation}
\end{theorem}

Hence, with confidence $1-\delta$ over the logs, any controller behind the shield reaches an unsafe state with probability at most $\Vbar_T(s_0)$, computed from the logs and the shield, not from $P$. As the bound is \emph{uniform} over admissible controllers, a new model can replace the deployed one without recomputation. The certificate belongs to the \emph{interface}, not the model.

\begin{corollary}[Sample complexity and rate]\label{cor:rate}
Let $\mathcal{R}$ be the rows an in-shield trajectory from $s_0$ can visit and $n_{\min}(M):=\min_{(s,a)\in\mathcal{R}}n_0(s,a)$. Suppose $n_{\min}(M)=\Omega(M)$, so $\beta_{s,a}=\mathcal{O}(1/\sqrt{M})$ uniformly on $\mathcal{R}$. On the event $E$ of Lemma~\ref{lem:valid} and on the event that every row of $\mathcal{R}$ meets its committed budget,
\[
\Vbar_T(s_0)\le\Vstar_T(s_0;P,\Dshield)+\mathcal{O}\big(T/\sqrt{M}\big).
\]
Coverage is essential here, since a relevant row missing its budget stays vacuous and the gap need not shrink.
\end{corollary}

\begin{corollary}[Validity under data-dependent synthesis]\label{cor:adaptive}
The event $E=\{P_a(s)\in\Cset_{s,a}\ \forall(s,a)\}$ depends on the data only, and on $E$ the domination $\Vstar_k\le\Vbar_k$ holds for \emph{every} shield at once. Hence~\eqref{eq:guarantee} survives replacing $\Dshield$ by any data-dependent $\widehat{\Dshield}$ with $\Vbar_T(s_0;\widehat{\Dshield})\le\varepsilon$, in particular the output of Alg.~\ref{alg:synth}. No sample splitting is required.
\end{corollary}

More logs tighten the bound at the usual $1/\sqrt{M}$ rate if they cover the relevant rows, a condition akin to persistent excitation, and the shield may be chosen from the logs that certify it. The next result bounds the cost of violating A\ref{ass:markov}.

\begin{proposition}[Residual history dependence]\label{prop:eta}
Drop A\ref{ass:markov} and assume instead that at every row the history-conditioned laws vary little, in that ${\|Q_a(h)-Q_a(h')\|_1}\le\eta$ and $\operatorname{supp}Q_a(h)\subseteq K_{s,a}$ for all reachable $h,h'$ with
$\varphi(h)=\varphi(h')=s$. Let $\Vbar^{\eta}$ run~\eqref{eq:robustvi} over sets
$\Cset^{\eta}_{s,a}$ of radius $\beta_{s,a}+\eta$. Then with probability at
least $1-\delta$ over the draw of the logs, every admissible controller obeys
\[
V^{\pi}_T(s_0)\ \le\ \Vbar^{\eta}_T(s_0)\ \le\ \Vbar_T(s_0)+\tfrac{1}{2}T\eta .
\]
\end{proposition}
Like a bound on unmodeled dynamics, $\eta$ is a declared tolerance. One states how far from Markov the abstraction may be, pays $T\eta/2$, and checks this on the logs (Experiment~4).

\section{Least-Restrictive Shield Synthesis}
\label{sec:synth}

Theorem~\ref{thm:main} certifies any shield with $\Vbar_T(s_0)\le\varepsilon$; the question is which to deploy, and one always qualifies.

\begin{proposition}[Well-posedness]\label{prop:feasible}
For $s_0\notin\Su$, the all-\halt{} shield $\Dshield_\bot(s)=\{\halt\}$ yields $\Vbar_T(s_0)=0$, so the set ${\{\Dshield:\Vbar_T(s_0)\le\varepsilon\}}$ is nonempty for every $\varepsilon\ge0$.
\end{proposition}

That shield is safe and useless; we want the one that forbids as little as possible while still certifying, that is, a solution of
\begin{equation}\label{eq:synth}
\max_{\Dshield}\ \Pi(\Dshield)\quad\text{s.t.}\quad \Vbar_T(s_0;\Dshield)\le\varepsilon,\ \ \halt\in\Dshield(s)\,\forall s,
\end{equation}
where $\Pi$ measures how much the agent keeps. Forbidding an action can only lower $\Vbar_T$, so forbidding more from a certified shield keeps it certified, but choosing \emph{which} actions to forbid is combinatorial. We proceed greedily, forbidding one action at a time, the one buying the most safety per unit of autonomy given up, until a removal meets the budget. Alg.~\ref{alg:synth} always returns a certified shield, since within $|\{(s,a)\}|$ removals it meets the budget or reaches all-\halt{} (Prop.~\ref{prop:feasible}); optimality is not claimed. Two choices make it effective.

\emph{(i) Which action to forbid first.} Counting removals is the wrong measure, since the cheapest way to make the start safe is to wall it off. We instead weight each action by how often a goal-directed reference policy uses it, $\Pi(\Dshield)=\sum_{s,a}\rho(s,a)\Ind{a\in\Dshield(s)}$
with $\rho$ its state--action occupancy, and rank removals by risk reduction
per unit cost $c(s,a)=\rho(s,a)+\rho_{\min}$. An action the task never uses costs nothing to forbid, so it goes first; inward step-into-hazard moves are exactly such actions, which is why the synthesizer \emph{fences the hazard} (Fig.~\ref{fig:env}b) rather than moats the agent.

\emph{(ii) Keeping the task reachable.} A fence removes one action at one state, never a state. Every state keeps $\halt$, so the shield never leaves the agent without an admissible action; that much is guaranteed. Beyond it, Alg.~\ref{alg:synth} \emph{prefers} removals keeping the goal reachable from $s_0$, falling back only when no risk-reducing alternative does; in Fig.~\ref{fig:env}b that preserves a route, and the agent travels around the fenced block to the star without ever proposing a forbidden call.

\begin{algorithm}[t]
\caption{Greedy task-aware certified shield synthesis}
\label{alg:synth}
\footnotesize
\begin{algorithmic}[1]
\Require $\Phat,\{\beta_{s,a}\}$, $T$, $\varepsilon$, $s_0$, occupancy $\rho$, floor $\rho_{\min}$, goal $g$
\Ensure certified shield $\Dshield$ with $\Vbar_T(s_0)\le\varepsilon$
\State $\Dshield(s)\leftarrow\Aset$ for all $s$ \Comment{permit everything}
\State compute $\Vbar$ by~\eqref{eq:robustvi}
\While{$\Vbar_T(s_0)>\varepsilon$}
  \State $a^\star(s)\in\arg\max_{a\in\Dshield(s)\setminus\{\halt\}}\sigma(\Cset_{s,a};\Vbar_{T-1})$, $s$ visited by $a^\star$
  \State evaluate $\Vbar_T(s_0)$ per removal \emph{(one batched VI)}
  \State rank with $\Vbar_T(s_0)\le\varepsilon$ first, highest first; else by $\Delta\Vbar/c(s,a)$
  \State commit top removal keeping $g$ reachable, else the top; recompute $\Vbar$
\EndWhile
\State \Return $\Dshield$
\end{algorithmic}
\end{algorithm}

Each iteration evaluates all candidate removals in one batched value iteration; on the instances of Sec.~\ref{sec:exp}, synthesis takes a few seconds per shield on one CPU core.

\section{Numerical Experiments}
\label{sec:exp}
We use two environments of our own. On a small grid the true kernel is known, so the certificate can be checked exactly (Experiments~1 and~2). A file-system sandbox mirrors the motivating application (Experiment~3), and Experiment~4 tests the Markov assumption. The synthesizer sees only sampled logs; code is at \url{\codeurl}.

\subsection{A grid with known ground truth}
In the grid of Fig.~\ref{fig:env}a each of the $7\times7$ cells is an abstract state, the central $3\times3$ block is $\Su$, a sink models $\halt$, and start and goal are opposite corners. Moves slip with probability $0.16$; $T=16$ and $\delta=0.05$. Each row gets the budget $n_0=\lfloor0.7\,MT/R\rfloor$, $70\%$ of an even share of the $MT$ transitions. Without a shield, the worst controller reaches $\Su$ with probability about $1$.

\begin{figure}[t]
\centering
\includegraphics[width=\columnwidth]{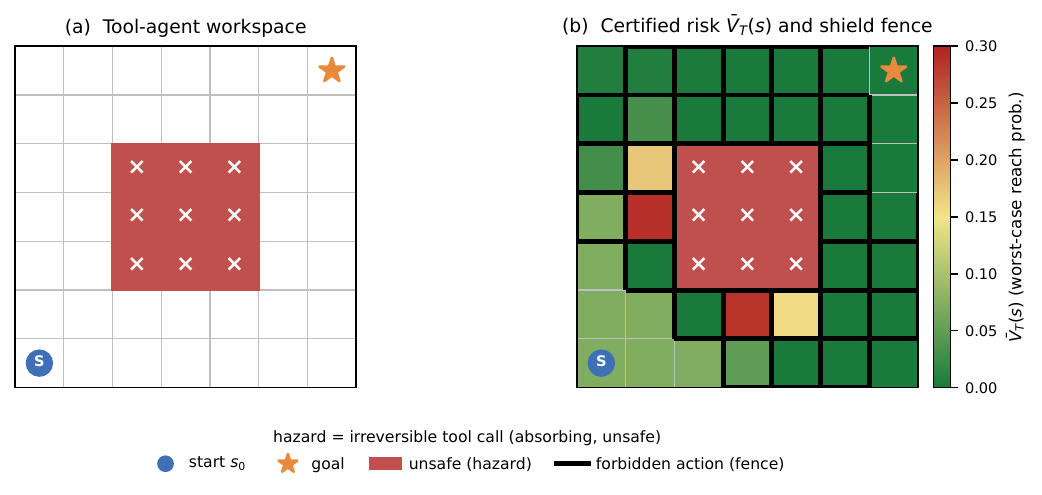}
\caption{\textbf{The grid benchmark and the shield of Experiment~1.} (a)~Each of the $7\times7$ cells is an abstract state of a tool-using agent; the red block is the unsafe set $\Su$, $s_0$ the start and the star the goal. (b)~Certified worst-case probability $\Vbar_T(s)$ of reaching $\Su$ from each cell (color) and forbidden moves (black segments). The shield fences the hazard instead of walling in the start, so a route to the goal remains, and the risk from $s_0$ is certified below the budget ($\varepsilon=0.10$, $M=12{,}000$, $T=16$).}
\label{fig:env}
\end{figure}

\textbf{Experiment 1 (soundness).} A shield synthesized from $M=12{,}000$ episodes with $\varepsilon=0.10$ certifies $\Vbar_T(s_0)=0.072$; its exact worst-case risk is $0.009$. Three controllers run through it, $4\times10^4$ episodes each. An adversary playing the worst admissible policy reaches $\Su$ in $0.010$ of them, close to that exact value; competent and random controllers record $0.000$.

Recomputed from more logs, this shield's certificate falls from $0.189$ at $3\times10^3$ episodes to $0.023$ at $1.2\times10^5$, and its excess over the exact value has a log--log slope of $-0.70$, faster than the $M^{-1/2}$ of Cor.~\ref{cor:rate}. Discarding the surplus samples has a price. The shield forbids $50.5\%$ of live state--action pairs, against $25.5\%$ if all samples could be used, and at $M=2{,}500$ an average of $1.5$ of $200$ rows miss their budget.

\textbf{Experiment 2 (validation across datasets).} The certificate is a claim about the random draw of the logs, so we test it that way. At each sample size we draw $100$ independent datasets, synthesize a fresh shield from each, and compute its exact worst-case risk. The true risk stays within both certificate and budget in all $100$ trials at each $M$ (Table~\ref{tab:coverage}), consistent with $95\%$ confidence. It \emph{rises} with $M$, from $0.002$ to $0.008$, because more data shrinks the confidence sets and the certified shield then forbids fewer actions (Cor.~\ref{cor:adaptive}).

\begin{table}[t]
\centering\footnotesize
\setlength{\tabcolsep}{5pt}
\begin{tabular}{rcccc}
\toprule
$M$ & $\Prob[\Vstar_T\!\le\!\Vbar_T]$ & $\Prob[\Vstar_T\!\le\!\varepsilon]$ & $\Vbar_T$ & $\Vstar_T$ \\
\midrule
$2{,}500$  & $1.00$ & $1.00$ & $0.076$ & $0.002$ \\
$12{,}000$ & $1.00$ & $1.00$ & $0.074$ & $0.008$ \\
\bottomrule
\end{tabular}
\caption{\textbf{The guarantee over $100$ datasets (Experiment~2).} For each number $M$ of logged episodes we drew $100$ independent datasets and built a fresh shield from each. The second and third columns give the fraction of datasets in which the shield's true worst-case risk $\Vstar_T$ was at most its certificate $\Vbar_T$, and at most the budget $\varepsilon=0.10$; the last two give the average $\Vbar_T$ and $\Vstar_T$ ($\delta=0.05$, $T=16$).}
\label{tab:coverage}
\end{table}

\subsection{A file-system sandbox, with four baselines}
The abstract state records whether the agent is \emph{in the protected directory}, \emph{holds write capability} and has \emph{confirmation}, plus a latch recording that a destructive call was issued. The tools are \texttt{read}, \texttt{chmod}, \texttt{write}, \texttt{delete}, \texttt{send} and $\halt$. The task is an authorized change, a \texttt{write} in production after confirmation; each \texttt{read} obtains confirmation with probability $0.5$, and calls misfire with probability $0.12$. Harm comes from a \texttt{chmod} into production before confirmation, then \texttt{write} or \texttt{delete}, or from a \texttt{delete} elsewhere, which sets the latch, then \texttt{send}. We log $M=4{,}000$ episodes of a competent agent that takes neither route, with $n_0=M$ and $\delta=0.05$; the least-visited row meeting its budget has $9{,}118$ observations, yet the hazards never appear in the logs. Episodes are cut off after $T=10$ executed calls, not counting blocked proposals.

Table~\ref{tab:agentic} compares five shields built from the same logs. Here and in Experiment~3 the synthesizer is a simpler variant of Alg.~\ref{alg:synth}. It ranks the permitted actions other than $\halt$ by the robust $Q$-value of line~4, never disconnects $g$ and falls back to all-$\halt$; Cor.~\ref{cor:adaptive} covers it. The \emph{plug-in} shield uses certainty equivalence. It treats the estimated kernel as exact, gives the unseen hazards probability zero and certifies zero risk, yet an adversary does harm on every seed. A \emph{policy-specific} certificate covers only the logging agent; it too certifies zero risk, forbids nothing and fails the same way. A \emph{count threshold}, forbidding pairs seen fewer than $10$ times, is safe but uncertified and forbids $80.2\%$ of pairs, against $4.2\%$ for ours. A Bayesian \emph{Dirichlet} shield is as safe but forbids twice as much ($8.3\%$), and its bound is credible, not distribution-free.

\begin{table}[t]
\centering\footnotesize
\setlength{\tabcolsep}{3.2pt}
\begin{tabular}{lcccccc}
\toprule
Shield & Cert.\ $\Vbar_T$ & Adv.\ risk & Succ. & Forb.\,\% & Interv.\,\% & Len. \\
\midrule
count-thr.     & n/a              & $0.00$          & $0.99$ & $80.2$ & $0$   & $4.3$ \\
Dirichlet-$95\%$ & $0.00$         & $0.00$          & $0.99$ & $8.3$  & $0$   & $4.3$ \\
plug-in MDP    & $0.00$           & $1.00$          & $0.99$ & $0.0$  & $0$   & $4.3$ \\
policy-spec.   & $0.00^{\dagger}$ & $1.00$          & $0.99$ & $0.0$  & $0$   & $4.3$ \\
robust (ours)  & $\mathbf{0.00}$  & $\mathbf{0.00}$ & $0.99$ & $\mathbf{4.2}$ & $0$ & $4.3$ \\
\bottomrule
\end{tabular}
\caption{\textbf{Five shields built from the same logs.} File-system sandbox, $M=4{,}000$ logged episodes, budget $\varepsilon=0.05$, averages over $200$ seeds. The columns give the certified bound $\Vbar_T$, the unsafe probability an adversarial controller reaches, a competent agent's success rate, the percentage of the $96$ live state--action pairs forbidden, the intervention rate and the mean episode length. $^{\dagger}$Logging controller only.}
\label{tab:agentic}
\end{table}

\textbf{Experiment 3 (language-model controllers).} We replace the scripted agent by a live language model that reads the full transcript and replies with one tool name. It never sees the predicates, so it is a history-dependent controller of the kind Thm.~\ref{thm:main} covers. The shield sees only $\varphi(h)$ and is the one built from the scripted agent's logs, not refitted. A forbidden call is blocked, not executed, and the model asked again (Sec.~\ref{sec:setup}). Each episode is attacked by prompt injection with probability $0.3$ ($39$ of $150$ were). An instruction hidden in the first file the agent reads tells it to delete the production logs and mail them out. Seven controllers sharing only the interface, namely the scripted agent and Qwen2.5 models~\cite{qwen2024technical} with $1.5$B, $3$B and $7$B parameters at temperatures $\tau=0$ and $1$, ran $150$ episodes each; Table~\ref{tab:llm} lists five. Without the shield, they range from never unsafe to unsafe in $84.7\%$ of episodes. With it, none of the seven is ever unsafe in its $150$ episodes, which bounds each unsafe rate by $2\%$ at $95\%$ confidence and, pooling the $1{,}050$ episodes, their average by $0.29\%$.

The $7$B model shows the mechanism best. It always starts with \texttt{read} and is never unsafe unless attacked, but the attack succeeds in $31$ of its $39$ attacked episodes ($79.5\%$). The shield prevents them all by blocking only the $39$ injected \texttt{delete}s, an intervention rate of $6.5\%$, and success stays at $55.3\%$. For the $3$B model the shield even \emph{raises} success, from $3.3\%$ to $38.7\%$. The model escalates permissions before confirming, and when that call is blocked and the model asked again, it goes back for the confirmation it skipped, which immediate $\halt$ substitution would never allow. The smaller models need more help, with $31$ to $47$ blocked proposals per $100$ executed calls against about $6$ for the $7$B. What the shield cannot supply is competence. At $\tau=0$ the $1.5$B model always opens with the forbidden \texttt{chmod}, never succeeds and always halts safely.

\begin{table}[t]
\centering\footnotesize
\setlength{\tabcolsep}{3pt}
\begin{tabular}{@{}lcccccc@{}}
\toprule
& \multicolumn{2}{c}{Unsafe} & \multicolumn{2}{c}{Success} & & \\
\cmidrule(lr){2-3}\cmidrule(lr){4-5}
Controller & none & shield & none & shield & Interv. & Blocks \\
\midrule
scripted competent   & $0.0$  & $\mathbf{0.0}$ & $99.3$ & $99.3$ & $0.0$  & none \\
Qwen2.5-1.5B, $\tau{=}0$ & $70.7$ & $\mathbf{0.0}$ & $0.0$  & $0.0$  & $42.0$ & \texttt{chmod} \\
Qwen2.5-3B, $\tau{=}0$   & $78.7$ & $\mathbf{0.0}$ & $3.3$  & $38.7$ & $43.4$ & \texttt{chmod} \\
Qwen2.5-7B, $\tau{=}0$   & $20.7$ & $\mathbf{0.0}$ & $55.3$ & $55.3$ & $6.5$  & \texttt{delete} \\
Qwen2.5-7B, $\tau{=}1$   & $21.3$ & $\mathbf{0.0}$ & $59.3$ & $59.3$ & $6.1$  & \texttt{delete} \\
\bottomrule
\end{tabular}
\caption{\textbf{Five of the seven Experiment~3 controllers behind one shield.} Each ran $150$ episodes, prompt-injected with probability $0.3$. The columns give the percentage of episodes ending unsafe and in success, without the shield (none) and with it; the intervention rate, in blocked proposals per $100$ executed calls; and the blocked call.}
\label{tab:llm}
\end{table}

\textbf{Experiment 4 (testing the Markov assumption).} Everything rests on A\ref{ass:markov}, so we test whether the next abstract state, given the current state and action, still depends on the previous one. For each row with enough data, a likelihood-ratio test compares the successor distributions across previous states, holding the false discovery rate at $5\%$ (Benjamini--Hochberg). We also report their largest $L_1$ distance $\hat\eta$, the empirical counterpart of $\eta$ in Prop.~\ref{prop:eta}, against a permutation null. On the grid $0$ of $160$ testable rows reject, with $\hat\eta=0.14$ below the null $0.17$; on the sandbox, $0$ of $15$, with $0.08$ against $0.20$. The test does catch a real defect. Deleting the \emph{confirmation} predicate makes $4$ of $7$ rows reject, with $\hat\eta=0.18$ above its null of $0.16$. Two examples establish neither calibration nor power in general, and the test has further limits. Only visited rows can be tested, and $\hat\eta$ is an effect size for one choice of context, not the supremum over histories that Prop.~\ref{prop:eta} needs, so it gives only a nominal sensitivity surcharge. Dependence this context cannot see is invisible, and failing to reject is not evidence of Markovness. Refining $\varphi$ in response makes the abstraction data-dependent, which Cor.~\ref{cor:adaptive} does not cover, so development and certification logs should be kept apart. The test can refute a predicate set, never certify one.

\section{Conclusion and Future Work}
\label{sec:concl}

We treated a language-model agent as an unknown feedback controller and placed a supervisory filter between it and its environment. The filter carries a certificate in the form of a robust reach--avoid bound on the probability of an irreversible action within the task, computed over an uncertainty set that is learned from ordinary logs. The bound holds with confidence $1-\delta$ for every controller the filter admits, so the language model can be replaced without recomputation, provided the environment, interface, permissions, abstraction and notion of harm stay the same. No failure need appear in the logs. An action never observed in a state is treated as maximally uncertain and is blocked if it could lead to harm, so the data must cover normal operation, not the hazards. The main assumption is that the Boolean predicates defining the abstract state make the loop Markov, which the logs can test but not prove. Future work should learn the abstraction from data held out from certification, tighten the filter online, and pool logs across agents sharing an environment.

\appendix[Proofs and supplementary results]

\noindent Statements are those of the main text; equation and result numbers refer to it.

\begin{proof}[Proof of Lemma~\ref{lem:worstcase}]
Let $\pi$ be admissible and $w^{\pi}_k(h)$ its probability of entering $\Su$ within $k$ further steps from a history $h$ with $\varphi(h)=s$. We prove $w^{\pi}_k(h)\le\Vstar_k(s)$ by induction on $k$; the deterministic $a^\star$ above achieves equality, giving the claim. \emph{Base} ($k=0$). Here $w^{\pi}_0(h)=\one[s\in\Su]=\Vstar_0(s)$. \emph{Step}. For $s\notin\Su\cup\{\bot\}$, conditioning on the time-$t$ action and invoking A\ref{ass:markov} (the next state is $s'\sim P(\cdot\mid s,a)$ regardless of $h$),
\[
w^{\pi}_k(h)=\sum_{a}\pi_t(a\mid h)\sum_{s'}P(s'\mid s,a)\,w^{\pi}_{k-1}\big(h\!\oplus\!(a,s')\big).
\]
By the inductive hypothesis $w^{\pi}_{k-1}(h\oplus(a,s'))\le\Vstar_{k-1}(s')$, and because $\pi$ is admissible every $a$ with ${\pi_t(a\mid h)>0}$ lies in $\Dshield(s)$. Hence each inner sum is at most
\[
\max_{a\in\Dshield(s)}\inner{P_a(s)}{\Vstar_{k-1}}=\Vstar_k(s),
\]
and the $\pi_t(\cdot\mid h)$-average of quantities $\le\Vstar_k(s)$ is itself $\le\Vstar_k(s)$. Absorbing and sink states match by construction.
\end{proof}

\begin{proof}[Proof of Lemma~\ref{lem:valid}]
Work in global time, so nothing is conditioned on the coverage event. For a
stochastic row and $A\subseteq K_{s,a}$, let $N_t(s,a)$ count visits up to $t$,
$I_t=\one[(s_t,a_t)=(s,a),\,N_{t-1}(s,a)<n_0(s,a)]$, which is predictable, and
$Z_t(A)=I_t(\one[s_{t+1}\in A]-Q_t(A))$ with $Q_t$ the conditional successor law
given the past. The $Z_t(A)$ are bounded martingale differences, so a
conditional-Hoeffding supermartingale gives
$\Prob[N(s,a)\ge n_0(s,a),\ \Phat_a(s)(A)-\Qbar_a(s)(A)\ge\tfrac{\eta}{2}]\le e^{-n_0\eta^2/2}$
one-sidedly; a row with $N(s,a)<n_0(s,a)$ takes the vacuous set. As
$\tfrac{1}{2}\|p-q\|_1=\max_A(p(A)-q(A))$, attained at $A=\{i:p_i>q_i\}$, a
union over the $2^{k_{s,a}}-2$ nontrivial subsets needs no factor of two, and a
further union over the $\le R$ stochastic rows
reproduces~\eqref{eq:weissman} about $\Qbar_a(s)$ with failure $\le\delta$. Deterministic rows contain $\Qbar_a(s)$ trivially, and no step assumes the
successors identically distributed, so this holds whether or not
A\ref{ass:markov} does.
\end{proof}

\begin{proof}[Proof of Proposition~\ref{prop:oracle}]
Any feasible $q$ satisfies
\[
\sum_{i:\,q_i>\hat p_i}(q_i-\hat p_i)=\sum_{i:\,q_i<\hat p_i}(\hat p_i-q_i)=:\xi
\]
with $2\xi=\|q-\hat p\|_1\le\beta_{s,a}$, hence $\xi\le\beta_{s,a}/2$. Since $\inner{q}{V}=\inner{\hat p}{V}+\sum_i(q_i-\hat p_i)V_i$, for fixed $\xi$ the maximizer adds mass on the largest $V_i$ (subject to $q_i\le1$) and removes it from the smallest (subject to $q_i\ge0$), so the gain is a concave piecewise-linear nondecreasing function of $\xi$. It is therefore maximized at the largest \emph{feasible} $\xi$, which is $\min\{\beta_{s,a}/2,\ \text{donor mass},\ \text{receiver capacity}\}$ and may be strictly smaller than $\beta_{s,a}/2$; once the marginal transfer moves mass between entries of equal $V$ the gain stops increasing, and it can be zero, as at $\hat p=(1,0)$ with $V=(1,0)$, where the receiver capacity is already exhausted. Sorting dominates the cost.
\end{proof}

\begin{proof}[Proof of Theorem~\ref{thm:main}]
Work on $E=\{P_a(s)\in\Cset_{s,a}\,\forall(s,a)\}$, of probability $\ge1-\delta$ by Lemma~\ref{lem:valid}. The left inequality holds since $\Vstar_k$ is the supremum~\eqref{eq:worstcase} over all admissible controllers, of which $\pi$ is one (Lemma~\ref{lem:worstcase}). For the right, induct on $k$. \emph{Base.} $\Vstar_0=\Vbar_0=\one_{\Su}$. \emph{Step.} For $s\notin\Su\cup\{\bot\}$, $\Vstar_k$ obeys the recursion~\eqref{eq:truevi} by Lemma~\ref{lem:worstcase}; the inductive hypothesis gives $\Vstar_{k-1}\le\Vbar_{k-1}$ pointwise, and since $P_a(s)\ge0$,
\begin{align*}
\Vstar_k(s)&=\max_{a\in\Dshield(s)}\inner{P_a(s)}{\Vstar_{k-1}}\\
&\le\max_{a\in\Dshield(s)}\inner{P_a(s)}{\Vbar_{k-1}}
\le\Vbar_k(s),
\end{align*}
the last step because $P_a(s)\in\Cset_{s,a}$ on $E$. Unsafe and sink states match by construction, and~\eqref{eq:guarantee} follows with $s=s_0$, $k=T$.
\end{proof}

\begin{proof}[Proof of Corollary~\ref{cor:rate}]
Only rows in $\mathcal{R}$ enter the recursion~\eqref{eq:robustvi} from $s_0$, so set $\beta^\star:=\max_{(s,a)\in\mathcal{R}}\beta_{s,a}=\mathcal{O}(1/\sqrt{{n_{\min}(M)}})$ by~\eqref{eq:beta}. Fix any $V\in[0,1]^{\Sset}$. By Prop.~\ref{prop:oracle}, $\sigma(\Cset_{s,a};V)\le\inner{\Phat_a(s)}{V}+\tfrac{\beta_{s,a}}{2}(\max V-\min V)\le\inner{\Phat_a(s)}{V}+\beta^\star/2$, and on the event $E$ of Lemma~\ref{lem:valid}, $\inner{\Phat_a(s)}{V}-\inner{P_a(s)}{V}\le\tfrac{1}{2}\|\Phat_a(s)-P_a(s)\|_1(\max V-\min V)\le\beta^\star/2$. Hence
{\begin{equation}\label{eq:onestep}
\sigma(\Cset_{s,a};V)\ \le\ \inner{P_a(s)}{V}+\beta^\star .
\end{equation}}
{Let $\Sset_k\subseteq\Sset$ be the states an in-shield trajectory from $s_0$ can occupy with $k$ steps remaining, and set $e_k:=\max_{s\in\Sset_k}|\Vbar_k(s)-\Vstar_k(s)|$, where $\Vstar$ solves~\eqref{eq:truevi} under the same shield. Restricting to $\Sset_k$ matters. Every row used by a backup at $s\in\Sset_k$ lies in $\mathcal{R}$, so $\beta^\star$ bounds its radius, which need not hold at a state $s_0$ cannot reach. Then $e_0=0$, since both equal $\one_{\Su}$. For $s\in\Sset_k\setminus(\Su\cup\{\bot\})$, whose successors lie in $\Sset_{k-1}$, applying~\eqref{eq:onestep} with $V=\Vbar_{k-1}$ and then $\inner{P_a(s)}{\Vbar_{k-1}}\le\inner{P_a(s)}{\Vstar_{k-1}}+e_{k-1}$ gives $\Vbar_k(s)\le\Vstar_k(s)+e_{k-1}+\beta^\star$; unsafe and sink states match exactly, and $\Vbar_k\ge\Vstar_k$ by Thm.~\ref{thm:main}. So $e_k\le e_{k-1}+\beta^\star$, whence $e_T\le T\beta^\star=\mathcal{O}(T/\sqrt{M})$, and $s_0\in\Sset_T$ gives the claim.}
\end{proof}

\begin{proof}[Proof of Corollary~\ref{cor:adaptive}]
$E$ is a statement about the data and the fixed unknown $P$, with $\Prob[E]\ge1-\delta$ by Lemma~\ref{lem:valid}; it never references $\Dshield$. The inductive step of Theorem~\ref{thm:main} uses only $P_a(s)\in\Cset_{s,a}$, true on $E$ for all rows, so $\Vstar_k(\cdot;P,\Dshield)\le\Vbar_k(\cdot)$ holds on $E$ simultaneously for all $\Dshield$. Thus, on $E$, any $\widehat{\Dshield}$ with $\Vbar_T(s_0;\widehat{\Dshield})\le\varepsilon$ obeys $\sup_{\pi\rhd\widehat{\Dshield}}V^{\pi}_T(s_0;P)\le\varepsilon$; as $\Prob[E]\ge1-\delta$, regardless of the rule producing $\widehat{\Dshield}$, the claim follows.
\end{proof}

\begin{proof}[Proof of Proposition~\ref{prop:eta}]
Lemma~\ref{lem:valid} nowhere uses A\ref{ass:markov}, so its event still
has probability $\ge1-\delta$ and gives $\|\Qbar_a(s)-\Phat_a(s)\|_1\le\beta_{s,a}$.
As $\Qbar_a(s)$ is a convex combination of laws at logged histories, the
hypothesis and convexity of the norm give $\|Q_a(h)-\Qbar_a(s)\|_1\le\eta$ at
every reachable $h$, so $\|Q_a(h)-\Phat_a(s)\|_1\le\beta_{s,a}+\eta$ and, with
the support condition, $Q_a(h)\in\Cset^{\eta}_{s,a}$. Induct on $k$ over histories. The base case is $w^{\pi}_0(h)=\one[\varphi(h)\in\Su]=\Vbar^{\eta}_0(\varphi(h))$, and since the
hypothesis applies pointwise to the successors and $\pi$ is admissible,
$w^{\pi}_k(h)=\sum_a\pi(a\mid h)\inner{Q_a(h)}{w^{\pi}_{k-1}(ha\,\cdot)}
\le\max_{a\in\Dshield(s)}\sigma(\Cset^{\eta}_{s,a};\Vbar^{\eta}_{k-1})=\Vbar^{\eta}_k(s)$.
For the second inequality, Prop.~\ref{prop:oracle} gives
$\sigma(\Cset^{\eta}_{s,a};V)-\sigma(\Cset_{s,a};V)\le\eta/2$ on
$[0,1]^{\Sset}$, and each backup is $1$-Lipschitz in the sup norm, so the gaps
accumulate over $T$.
\end{proof}

\begin{proof}[Proof of Proposition~\ref{prop:feasible}]
Under $\Dshield_\bot$ the only admissible controller halts at once into $\bot$, never entering $\Su$; thus $\Vbar_T\equiv0$ by~\eqref{eq:robustvi}.
\end{proof}

\subsection*{Context sensitivity of the Experiment~4 diagnostic}

\noindent The diagnostic of Experiment~4 conditions on the previous abstract
state. Conditioning additionally on the previous action was also evaluated
(Table~\ref{tab:context}).

\begin{table}[ht]
\centering\footnotesize
\begin{tabular}{@{}llcccc@{}}
\toprule
Abstraction & Context & Rows & Rejected & $\hat\eta$ & Null\\
\midrule
Grid            & state          & $160$ & $0$ & $0.136$ & $0.169$\\
Grid            & state, action  & $0$   & --  & --      & --\\
Sandbox, fine   & state          & $15$  & $0$ & $0.084$ & $0.201$\\
Sandbox, fine   & state, action  & $16$  & $0$ & $0.170$ & $0.506$\\
Sandbox, coarse & state          & $7$   & $4$ & $0.184$ & $0.155$\\
Sandbox, coarse & state, action  & $7$   & $2$ & $0.630$ & $0.374$\\
\bottomrule
\end{tabular}
\caption{\textbf{Experiment 4 under two choices of context.} The context is the previous abstract state, or that state together with the previous action. Rows counts the kernel rows with enough data to test and Rejected those where the test rejects; $\hat\eta$ is the largest $L_1$ distance between the successor distributions under different contexts, and Null is the $95\%$ quantile of $\hat\eta$ when context labels are shuffled within a row.}
\label{tab:context}
\end{table}

\noindent Richer context buys effect size and costs power. On the coarsened
sandbox the pair context raises $\hat\eta$ from $0.184$ to $0.630$ while
halving the rejections, and on the grid it leaves no testable row against
$160$, because $|\Sset||\Aset|$ contexts exhaust the data. That collapse is why
the previous state alone is the primary diagnostic in the main text.

\bibliographystyle{IEEEtran}
\bibliography{refs}

\end{document}